\documentclass[letterpaper,11pt]{article}
\usepackage[margin=0.8in]{geometry}

\usepackage{ifpdf}
\ifpdf
    \usepackage[pdftex]{graphicx}
    \usepackage[update]{epstopdf}
\else
	\usepackage{graphicx}
\fi
\usepackage{wrapfig,bbm}
\usepackage{enumitem,color}
\usepackage{multirow}
\usepackage{amsthm}
\usepackage{subfigure}
\usepackage{setspace}

\newtheorem{thm}{Theorem}

\newtheorem{lemma}{Lemma}
\newtheorem{corollary}{Corollary}

\newtheorem{defn}{Definition}

\usepackage[mathscr]{euscript}
\usepackage{eufrak}
\usepackage{bbold}
\usepackage{array}
\usepackage[cmex10]{amsmath}
\usepackage{amssymb}
\usepackage{amstext}
\usepackage{cite,setspace}
\usepackage{algorithm2e}
\usepackage{url}
\usepackage[dvipsnames]{xcolor}

\newcommand{\Expt}{\mbox{${\mathbb E}$} }

\DeclareGraphicsExtensions{.pdf,.png,.jpg,.eps}

\usepackage{mathtools}

\allowdisplaybreaks

\begin{document}
\title{A Shannon-Theoretic Approach to the Storage-Retrieval Tradeoff in PIR Systems}

\author{Chao Tian, Hua Sun, and Jun Chen}
\date{}
\maketitle


\begin{abstract}
We consider the storage-retrieval rate tradeoff in private information retrieval (PIR) systems using a Shannon-theoretic approach. Our focus is mostly on the canonical two-message two-database case, for which a coding scheme based on random codebook generation and the binning technique is proposed. This coding scheme reveals a hidden connection between PIR and the classic multiple description source coding problem. We first show that when the retrieval rate is kept optimal, the proposed non-linear scheme can achieve better performance over any linear scheme. Moreover, a non-trivial storage-retrieval rate tradeoff can be achieved beyond space-sharing between this extreme point and the other optimal extreme point, achieved by the retrieve-everything strategy. We further show that with a method akin to the expurgation technique, one can extract a zero-error PIR code from the random code. Outer bounds are also studied and compared to establish the superiority of the non-linear codes over linear codes. 
\end{abstract}
\section{Introduction}

Private information retrieval (PIR) addresses the situation of storing $K$ messages of $L$-bits each in $N$ databases, with the requirement that the identity of any requested message must be kept private from any one (or any small subset) of the databases. The early works were largely computer science theoretic \cite{chor1995private}, where $L=1$, and the main question is the scaling law of the retrieval rate in terms of $(K,N)$. 

The storage overhead in PIR systems has been studied in the coding and information theory community, from several perspectives using mainly two problem formulations.  Shah {\em et al.} \cite{Shah_Rashmi_Kannan} considered the problem when $N$ is allowed to vary with $L$ and $K$, and obtained some conclusive results. In a similar vein, for $L=1$, Fazeli {\em et al.}  \cite{Fazeli_Vardy_Yaakobi} proposed a technique to convert any linear PIR code to a new one with low storage overhead by increasing $N$. Other notable results along this line can be found in \cite{Rao_Vardy, blackburn2019pir, blackburn2019pirschemes, zhang2019private, vajha2017binary, asi2018nearly}.

An information theoretic formulation of the PIR problem was considered in \cite{Chan_Ho_Yamamoto}, where $L$ is allowed to increase, while $(N,K)$ are kept fixed. Important properties on the tradeoff between the storage rate and retrieval rate were identified in \cite{Chan_Ho_Yamamoto}, and a linear code construction was proposed. In this formulation, even without any storage overhead constraint, characterizing the minimum retrieval rate in the PIR systems is nontrivial, and this capacity problem was settled in \cite{sun2017PIRcapacity}. Tajeddine {\em et al.} \cite{Tajeddine_Rouayheb} considered the capacity problem when the message is coded across the databases with a maximum-distance separable (MDS) code, which was later solved by Banawan and Ulukus  \cite{banawan2018capacity}. Capacity-achieving code designs with optimal message sizes were given in \cite{tian2018capacity,zhou2020capacity}. Systems where servers can collude were considered in \cite{Sun_Jafar_TPIR}. There have been various extensions and generalizations, and the recent survey article \cite{ulukus2022private} provides a comprehensive overview on efforts following this information theoretic formulation.

In many existing works, the storage component and the PIR component are largely designed separately, usually by placing certain structural constraints on one of them; {\em e.g.,} the MDS coding requirement for the storage component \cite{banawan2018capacity}, or the storage is uncoded \cite{attia2020capacity}; moreover, the code constructions are almost all linear. The few exceptions we are aware of are \cite{sun2018multiround,sun2019breaking,guo2021new}. In this work, we consider the information theoretic formulation of the PIR problem, without placing any additional structural constraints on the two components, and explicitly investigate the storage-retrieval tradeoff \emph{region}. We mostly focus on the case $N=K=2$ here since it provides the most important intuition; we refer to this as the $(2,2)$ PIR system. Our approach naturally allows the joint design of the two components using either linear or \emph{non-linear} schemes.

The work in \cite{sun2018multiround} is of significant relevance to our work, where the storage overhead was considered in both single-round and multi-round PIR systems, when the retrieval rate must be optimal. Although multi-round PIR has the same capacity as single-round PIR, it was shown that  at the minimum retrieval rate,  a multi-round, $\epsilon$-error, non-linear code can indeed break the storage performance barrier of an optimal single-round, zero error, linear code. The question whether all the three differences are essential to overcome this barrier was left as an open question. 

In this work, we show that a non-linear code is able to achieve better performance than the optimal linear code in the single-round zero-error $(2,2)$ PIR system, over a range of the storage rates. This is accomplished by providing a Shannon-theoretic coding scheme based on random codebook generation and the binning technique. The proposed scheme at the minimum retrieval rate is conceptually simpler, and we present it as an explicit example. The general inner bound is then provided, and we show an improved tradeoff can be achieved beyond space-sharing between the minimum retrieval rate code and the other optimal extreme point. By leveraging a method akin to the expurgation technique, we further show that one can extract a zero-error deterministic PIR code from the random $\epsilon$-error PIR code. Outer bounds are also studied for both general codes and linear codes, which allow us to establish conclusively the superiority of non-linear codes over linear codes. Our work essentially answers the open question in \cite{sun2018multiround}, and shows that in fact only non-linearity is essential in breaking the aforementioned barrier. 

A preliminary version of this work was presented first in part in \cite{tian2018shannon}. In this updated article, we provide a more general random coding scheme, which reveals a hidden connection to the multiple description source coding problem \cite{gamal1982achievable}. Intuitively, we can view the retrieved message as certain partial reconstruction of the full set of messages, instead of a complete reconstruction of a single message. Therefore, the answers from the servers can be viewed as descriptions of the full set of messages, which are either stored directly at the servers or formed at the time of request, and the techniques seen in multiple description coding become natural in the PIR setting. Since the publication of the preliminary version \cite{tian2018shannon}, several subsequent efforts have been made in studying the storage-retrieval tradeoff in the PIR setting, which provided stronger and more general information theoretic outer bounds and several new linear code constructions \cite{sun2019breaking,tian2020storage,guo2021new}. However, the Shannon-theoretic random coding scheme given in \cite{tian2018shannon} remains the best performance for the $(2,2)$ case, which motivate us to provide the general coding scheme in this work and to make the connection to multiple description source coding more explicit. It is our hope that this connection may bring existing coding techniques for the multiple description problem to the study of the PIR problem. 

\section{Preliminaries}

The problem we consider is essentially the same as that in \cite{sun2017PIRcapacity}, with the additional consideration on the storage overhead constraint at the databases. We provide a formal problem definition in the more traditional Shannon-theoretic language, to facilitate subsequent treatment. Some relevant results on this problem are also reviewed briefly in this section. 

\subsection{Problem Definition}

There are two independent messages, denoted as $W_1$ and $W_2$, in this system, each of which is generated uniformly at random in the finite field $\mathbb{F}_{2}^L$, {\em i.e.,} each message is an $L$-bit sequence. There are two databases to store the messages, which are produced by two encoding functions operating on $(W_1,W_2)$
\begin{gather*}
\phi_n: \mathbb{F}_{2}^L\times\mathbb{F}_{2}^L\rightarrow \mathbb{F}_{2}^{\alpha_n},\quad n=1,2,
\end{gather*}
where $\alpha_n$ is the number of storage symbols at database-$n$, $n=1,2$, which is a deterministic function of $L$, i.e., we are using fixed length codes for storage. 
We write $S_1=\phi_1(W_1,W_2)$ and $S_2=\phi_2(W_1,W_2)$. When a user requests message-$k$, it generates two queries $(Q^{[k]}_1,Q^{[k]}_2)$ to be sent to the two databases, randomly in the alphabet $\mathcal{Q}\times\mathcal{Q}$. 
Note the joint distribution satisfies the condition 
\begin{align}
P_{W_1,W_2,Q^{[k]}_1,Q^{[k]}_2}=P_{W_1,W_2} P_{Q^{[k]}_1,Q^{[k]}_2},\quad k=1,2,
\end{align}
{\em i.e., } the messages and the queries are independent. The marginal distributions $P_{W_1,W_2}$ and $P_{Q^{[k]}_1,Q^{[k]}_2}$, $k=1,2$, thus fully specify the randomness in the system. 

After receiving the queries, the databases produce the answers to the query via a set of deterministic functions
\begin{align}
\varphi^{(q)}_{n}: \mathbb{F}_{2}^{\alpha_n}\rightarrow \mathbb{F}_{2}^{\beta^{(q)}_n},\quad q\in \mathcal{Q}, \, n=1,2.
\end{align}
We also write the answers $A_{n}^{[k]}=\varphi^{(Q^{[k]}_n)}_{n}(S_n)$, $n=1,2$. The user, with the retrieved information, wishes to reproduce the desired message through a set of decoding functions
\begin{align}
\psi^{(k,q_1,q_2)}: \mathbb{F}_2^{\beta^{(q_1)}_1}\times\mathbb{F}_2^{\beta^{(q_2)}_2}\rightarrow \mathbb{F}_{2}^L.
\end{align}
The outputs of the functions $\hat{W}_k=\psi^{(k,Q^{[k]}_1,Q^{[k]}_2)}(A_1^{[k]},A_2^{[k]})$ are essentially the retrieved messages.  We require the system to retrieve the message correctly (zero-error), {\em  i.e.}, $\hat{W}_k=W_k$ for $k=1,2$.

%


Alternatively, we can require the system to have a small error probability. Denote the average probability of coding error of a PIR code as
\begin{align}
&P_e=0.5\sum_{k=1,2}P_{W_1,W_2,Q^{[k]}_1,Q^{[k]}_2}(W_k\neq \hat{W}_k).\label{eqn:Pe}
\end{align}
An $(L,\alpha_1,\alpha_2,\beta_1,\beta_2)$ $\epsilon$-error PIR code is defined similar as a (zero-error) PIR code, except that the correctness condition is replaced by the condition that the probability of error $P_e\leq \epsilon$.

Finally, the privacy constraint stipulates that the identical distribution condition must be satisfied:
\begin{align}
P_{Q^{[1]}_n,A^{[1]}_n,S_n}=P_{Q^{[2]}_n,A^{[2]}_n,S_n},\quad n=1,2.
\end{align}
Note that one obvious consequence is that $P_{Q^{[1]}_n}=P_{Q^{[2]}_n}\triangleq P_{Q_n}$, for $n=1,2$. 

We refer to the code, which is specified by two probability distributions $P_{Q^{[k]}_1,Q^{[k]}_2}$, $k=1,2$, and a valid set of coding functions $\{\phi_n,\varphi_n^{(q)},\psi^{k,q_1,q_2}\}$ that satisfy both the correctness and privacy constraints, as an $(L,\alpha_1,\alpha_2,\beta_1,\beta_2)$ PIR code, where $\beta_n=\Expt_{Q_n}[\beta^{(Q_n)}_n]$, for $n=1,2$.

\begin{defn}
A normalized storage-retrieval rate pair $(\bar{\alpha},\bar{\beta})$ is achievable, if for any $\epsilon>0$ and sufficiently large $L$, there exists an $(L,\alpha_1,\alpha_2,\beta_1,\beta_2)$ PIR code, such that
\begin{gather}
L(\bar{\alpha}+\epsilon)\geq \frac{1}{2}(\alpha_1+\alpha_2), \,\, L(\bar{\beta}+\epsilon)\geq \frac{1}{2}\left(\beta_1+\beta_2\right).
\end{gather}
The collection of the achievable normalized storage-retrieval rate pair $(\bar{\alpha},\bar{\beta})$  is the achievable storage-retrieval rate region, denoted as $\mathcal{R}$. 
\end{defn}

 Unless explicitly stated, the rate region $\mathcal{R}$ is used for the zero-error PIR setting. In the definition above, we have used the average rates $(\bar{\alpha},\bar{\beta})$ across the databases instead of the individual rate vectors $\frac{1}{n}(\alpha_1,\alpha_2,\Expt_{Q_1}[\beta^{(Q_1)}_1],\Expt_{Q_2}[\beta^{(Q_2)}_2])$. This can be justified using the following lemma. 

\begin{lemma}
If an $(L,\alpha_1,\alpha_2,\beta_1,\beta_2)$ PIR code exists, then a $(2L, \alpha,\alpha,\beta,\beta)$ PIR code exists, where
\begin{align}
\alpha=\alpha_1+\alpha_2, \quad \beta = \beta_1+\beta_2.
\end{align}
\end{lemma}
This lemma can essentially be proved by a space-sharing argument, the details of which can be found in \cite{sun2018multiround}. The following lemma is also immediate using a conventional space-sharing argument.
\begin{lemma}
The region $\mathcal{R}$ is convex. 
\end{lemma}

\subsection{Some Relevant Known Results}

The capacity of a general PIR system with $K$ messages and $N$ databases is identified in \cite{sun2017PIRcapacity} as
\begin{align}
C=\frac{1-1/N}{1-1/N^K},
\end{align}
which in our definition corresponds to the case when $\bar{\beta}$ is minimized, and the proposed linear code achieves $(\bar{\alpha},\bar{\beta})=(K,(1-1/N^K)/(N-1))$. The capacity of MDS-code PIR systems was established in \cite{banawan2018capacity}. In the context of storage-retrieval tradeoff, this result can be viewed as providing the achievable tradeoff pairs
\begin{align}
(\bar{\alpha},\bar{\beta})=\left(t,\frac{1-t^K/N^K}{N-t}\right), t=1,2,\ldots,N.
\end{align}
However when specialized to the $(2,2)$ PIR problem, this does not provide any improvement over the space-sharing strategy between the trivial code of retrieval-everything and the code in \cite{sun2017PIRcapacity}.  By specializing the code in \cite{sun2017PIRcapacity}, it was shown in  \cite{sun2018multiround} that for the $(2,2)$ PIR problem, at the minimal retrieval value $\bar{\beta}=0.75$, the storage rate $\bar{\alpha}_{l}=1.5$ is achievable using a single-round, zero-error linear code, and in fact, it is the optimal storage rate that any single-round, zero-error linear code can achieve.

One of the key observations in \cite{sun2018multiround} is that a special coding structure appears to be the main difficulty in the $(2,2)$  PIR setting, which is illustrated in Fig. \ref{fig:skewedsymmetric}. Here message $W_1$ can be recovered from either $(X_1,Y_1)$ or $(X_2,Y_2)$, and message $W_2$ can be recovered from either $(X_1,Y_2)$ or $(X_2,Y_1)$; $(X_1,X_2)$ is essentially $S_1$ and is stored at database-1, and $(Y_1,Y_2)$ is essentially $S_2$ and is stored at database-2. It is clear that we can use the following strategy to satisfy the privacy constraint: when message $W_1$ is requested, with probability $1/2$, the user queries for either $(X_1,Y_1)$ or $(X_2,Y_2)$; for message 2, with probability $1/2$, the user queries for either $(X_1,Y_2)$ or $(X_2,Y_1)$. More precisely, the following probability distribution $P_{Q^{[1]}_1,Q^{[1]}_2}$ and $P_{Q^{[2]}_1,Q^{[2]}_2}$ can be used:
\begin{align}
P_{Q^{[1]}_1,Q^{[1]}_2}=\left\{
\begin{array}{ll}
0.5& (Q^{[1]}_1,Q^{[1]}_2)=(11)\\
0.5& (Q^{[1]}_1,Q^{[1]}_2)=(22)
\end{array}
\right.,\label{eqn:Q1}
\end{align}
and 
\begin{align}
P_{Q^{[2]}_1,Q^{[2]}_2}=\left\{
\begin{array}{ll}
0.5& (Q^{[2]}_1,Q^{[2]}_2)=(12)\\
0.5& (Q^{[2]}_1,Q^{[2]}_2)=(21)
\end{array}
\right..\label{eqn:Q2}
\end{align}

\begin{figure}[t!]
\centering
\includegraphics[width=5cm]{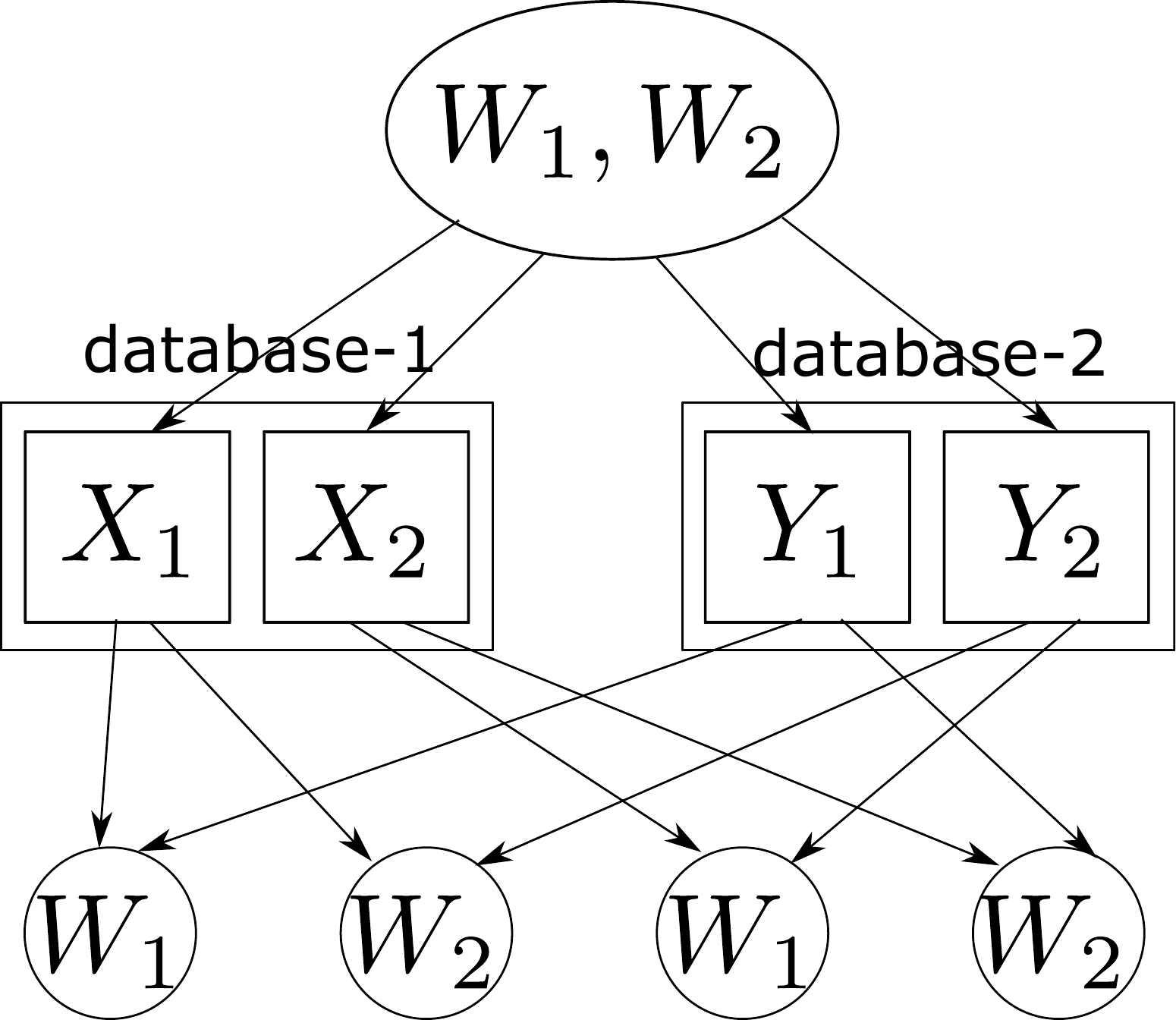}
\caption{A possible coding structure.\label{fig:skewedsymmetric}}
\end{figure}

\subsection{Multiple Description Source Coding}
The multiple description source coding problem \cite{gamal1982achievable} considers compressing a memoryless source $S$ into a total of $M$ descriptions, i.e., $M$ compressed bit sequences, such that the combinations of any subset of these descriptions can be used to reconstruct the source $S$ to guarantee certain quality requirements. The motivation of this problem is mainly to address the case when packets can be dropped randomly on a communication network. 

Denote the coding rate for each description as $R_i$, $i=1,2,\ldots,M$.  A coding scheme was proposed in \cite{venkataramani2003multiple}, which leads to the following rate region. Let $U_1,U_2,\ldots,U_M$ be $M$ random variables jointly distributed with $S$, then the following rates $(R_1,R_2,\ldots,R_M)$ and distortions $(D_{\mathcal{A}},\mathcal{A}\subseteq \{1,2,\ldots,M\})$ are achievable:
\begin{align}
\sum_{i\in \mathcal{A}} R_i&\geq \sum_{i\in \mathcal{A}} H(U_i)-H(\{U_i,i\in\mathcal{A}\}|S),\quad \mathcal{A}\subseteq \{1,2,\ldots,M\},\label{eqn:rates}\\
D_{\mathcal{A}}&\geq \Expt[d(S,f_{\mathcal{A}}(U_i,i\in\mathcal{A}))],\quad \mathcal{A}\subseteq \{1,2,\ldots,M\}.
\end{align}
Here $f_{\mathcal{A}}$ is a reconstruction mapping from the random variables $\{U_i,i\in\mathcal{A}\}$ to the reconstruction domain, $d(\cdot,\cdot)$ is a distortion metric that is used to measure the distortion, and $D_{\mathcal{A}}$ is the distortion achievable using the descriptions in the set $\mathcal{A}$. Roughly speaking, the coding scheme requires generating approximately $2^{nR_i}$ length-$n$ codewords in an i.i.d. manner using the marginal distribution $U_i$ for each $i=1,2,\ldots,M$, and the rate constraints insure that when $n$ is sufficiently large, with overwhelming probability there is a tuple of $M$ codewords $(u_1^n,u_2^n,\ldots,u_M^n)$, one in each codebook constructed earlier, that are jointly typical with the source vector $S^n$. In this coding scheme, the descriptions are simply the codeword indices of these codewords in these codebooks. For a given joint distribution $(S,U_1,U_2,\ldots,U_M)$, we refer to the rate region in (\ref{eqn:rates}) as the MD rate region $\mathcal{R}_{MD}(S,U_1,U_2,\ldots,U_M)$, and the corresponding random code construction the MD codebooks associated with $(S,U_1,U_2,\ldots,U_M)$.

The binning technique \cite{wyner1976rate} can be applied in the multiple description problem to provide further performance improvements, particularly when not all the combinations of the descriptions are required to satisfy certain performance constraints, but only a subset of them are; this technique has previously been used in \cite{pradhan2004n} and \cite{tian2010new} for this purpose. Assume that only the subsets of descriptions $\mathcal{A}_1,\mathcal{A}_2,\ldots,\mathcal{A}_T\subseteq \{1,2,\ldots,M\}$ have distortion requirements associated with the reconstructions using these descriptions, which are denoted as $D_{\mathcal{A}_i}$, $i=1,2,\ldots,T$. Consider the MD codebooks associated with $(S,U_1,U_2,\ldots,U_M)$ at rates $(R'_1,R'_2,\ldots,R'_M)\in\mathcal{R}_{MD}(S,U_1,U_2,\ldots,U_M)$, then assign the codewords in the $i$-th codebook uniformly at random into $2^{nR_i}$ bins with $0\leq R_i\leq R'_i$. The coding rates and distortions that satisfy the following constraints simultaneously for all $\mathcal{A}_i,i=1,2,\ldots,T$ are achievable: 
\begin{align}
\sum_{j\in \mathcal{J}} (R'_j-R_j)&\leq \sum_{j\in \mathcal{J}} H(U_j)-H\left(\left\{U_j,j\in\mathcal{J}\right\}\bigg{|}\left\{U_{j'},j'\in\mathcal{A}_i\setminus \mathcal{J}\right\}\right),\quad \forall\mathcal{J}\subseteq \mathcal{A}_i,\label{eqn:rates2}\\
D_{\mathcal{A}_i}&\geq \Expt[d(S,f_{\mathcal{A}_i}(U_j,j\in\mathcal{A}_i))].
\end{align}
We denote the collection of such rate vectors $(R_1,R_2,\ldots,R_M,R'_1,R'_2,\ldots,R'_M)$ as  $\mathcal{R}^*_{MD}((S,U_1,U_2,\ldots,U_M),(\{U_j,j\in\mathcal{A}_i\},i=1,2,\ldots,T))$, and refer to the corresponding codebooks as the MD$^*$ codebooks associated with the random variables $(S,U_1,U_2,\ldots,U_M)$ and the reconstruction sets $(\mathcal{A}_1,\mathcal{A}_2,\ldots,\mathcal{A}_T)$.

\section{A Special Case: Slepian-Wolf Coding for Minimum Retrieval Rate}
\label{sec:SW}

In this section, we consider the minimum-retrieval-rate case, and show that non-linear and Shannon-theoretic codes are beneficial.  We will be rather cavalier here and ignore some details, in the hope of better conveyance of the intuition. In particular, we ignore the asymptotic-zero probability of error that is usually associated with a random coding argument, but this will be addressed more carefully in Section \ref{sec:main}.

Let us rewrite the $L$-bit messages as 
\begin{align}
W_k=(V_k[1],\ldots,V_k[L])\triangleq V_k^L,\quad k=1,2.
\end{align}
The messages can be viewed as being produced from a discrete memoryless source $P_{V_1,V_2}=P_{V_1}\cdot P_{V_2}$, where $V_1$ and $V_2$ are independent uniform-distributed Bernoulli random variables. 

Consider the following auxiliary random variables 
\begin{gather}
X_1\triangleq V_1\land V_2,\quad X_2\triangleq(\neg V_1 )\land (\neg V_2),\nonumber\\
Y_1\triangleq V_1 \land (\neg V_2), \quad Y_2\triangleq (\neg V_1)\land V_2,\label{eqn:distribution}
\end{gather}
where $\neg$ is the binary negation, and $\land$ is the binary ``and'' operation. 
This particular distribution satisfies the coding structure depicted in Fig. \ref{fig:skewedsymmetric}, with $(V_1,V_2)$ taking the role of $(W_1,W_2)$, and the relation is non-linear. The same distribution was used in \cite{sun2018multiround}  to construct a multiround PIR code.  This non-linear mapping appears to allow the resultant code to be more efficient than linear codes. 

We wish to store $(X^L_1,X^L_2)$ at the first database in a lossless manner, however, store only certain necessary information regarding $Y^L_1$ and $Y^L_2$ to facilitate the recovery of $W_1$ or $W_2$. For this purpose, we will encode the message as follows:
\begin{itemize}
\item At database-1, compress  and store $(X^L_1,X^L_2)$ losslessly;
\item At database-2, encode $Y_1^L$ using a Slepian-Wolf code (or more precisely Sgarro's code with uncertainty side information \cite{sgarro1977source}), with either $X_1^L$ or $X^L_2$ at the decoder, whose resulting code index is denoted as $C_{Y_1}$; encode $Y^L_2$ in the same manner, independent of $Y_1^L$, whose code index is denoted as $C_{Y_2}$. 
\end{itemize}
It is clear that for database-1, we need roughly $\bar{\alpha}_1=H(X_1,X_2)$. At database-2, in order to guarantee successful decoding of the Slepian-Wolf code, we can chose roughly
\begin{align}
\bar{\alpha_2}&=\max(H(Y_1|X_1),H(Y_1|X_2))+\max(H(Y_2|X_1),H(Y_2|X_2))\nonumber\\
&=2H(Y_1|X_1),
\end{align}
where the second equality is due to the symmetry in the probability distribution. Thus we find that this code achieves
\begin{align}
\bar{\alpha}_{nl}&=0.5[H(X_1,X_2)+2H(Y_1|X_1)]\nonumber\\
&=0.75+0.75H(1/3,2/3)\nonumber\\
&=0.25+0.75\log_23\approx1.4387. 
\end{align}

The retrieval strategy is immediate from the coding structure in Fig. \ref{fig:skewedsymmetric}, with  $(V_1^L,V_2^L,X^L_1,X^L_2,C_{Y_1},C_{Y_2})$ serving the roles of $(W_1,W_2,X_1,X_2,Y_1,Y_2)$, and thus indeed the privacy constraint is satisfied. The retrieval rates are roughly as follows
\begin{gather}
\bar{\beta}_1^{(1)}=\bar{\beta}_1^{(2)}=H(X_1)=H(X_2),\\
\bar{\beta}_2^{(1)}=\bar{\beta}_2^{(2)}=H(Y_1|X_1),
\end{gather}
implying
\begin{gather*}
\bar{\beta}=0.5[H(X_1)+H(Y_1|X_1)]=0.5H(Y_1,X_1)=0.75. 
\end{gather*}

Thus at the optimal retrieval rate $\bar{\beta}=0.75$, we have
\begin{align}
\bar{\alpha}_{l}=1.5 \mbox { vs. }\bar{\alpha}_{nl}\approx 1.4387,
\end{align}
and clearly the proposed non-linear Shannon-theoretic code is able to perform better than the optimal linear code. We note that it was shown in \cite{sun2018multiround} by using a multround approach, the storage rate $\bar{\alpha}$ can be further reduced, however this issue is beyond the scope of this work.  In the rest of the paper, we build on the intuition in this special case to generalize and strengthen the coding scheme. 

\section{Main Result}
\label{sec:main}

\subsection{A General Inner Bound}
We first present a general inner bound to the storage-retrieval tradeoff region. Let $(V_1,V_2)$ be independent random variables uniformly distributed on $\mathbb{F}_2^{t}\times \mathbb{F}_2^{t}$. Define the region $\mathcal{R}^{(t)}_{in}$ to be the collection of $(\bar{\alpha},\bar{\beta})$ pairs for which there exist random variables $(X_0,X_1,X_2,Y_1,Y_2)$ jointly distributed with $(V_1,V_2)$ such that:
\begin{enumerate}
\item There exist deterministic functions $f_{1,1}$, $f_{1,2}$, $f_{2,1}$, and $f_{2,2}$ such that
\begin{gather}
V_1=f_{1,1}(X_0,X_1,Y_1)=f_{2,2}(X_0,X_2,Y_2),\quad 
V_2=f_{1,2}(X_0,X_1,Y_2)=f_{2,1}(X_0,X_2,Y_1)\label{eqn:generaldecoding4};
\end{gather}
\item There exist non-negative coding rates 
\begin{align}
&(\beta_1^{(0)},\beta_1^{(1)},\beta_1^{(2)},\beta_2^{(1)},\beta_2^{(2)},\gamma_1^{(0)},\gamma_1^{(1)},\gamma_1^{(2)},\gamma_2^{(1)},\gamma_2^{(2)})\notag\\
&\quad \in \mathcal{R}^*_{MD}\left(((V_1,V_2),X_0,X_1,X_2,Y_1,Y_2),\left(\{X_0,X_1,Y_1\},\{X_0,X_1,Y_2\},\{X_0,X_2,Y_1\},\{X_0,X_2,Y_2\}\right)\right);\label{eqn:10R}
\end{align}
\item There exist non-negative storage rates $(\alpha_1^{(0)},\alpha_1^{(1)},\alpha_1^{(2)},\alpha_2^{(1)},\alpha_2^{(2)})$ such that
\begin{gather}
\alpha_1^{(0)}\leq \beta_1^{(0)},\alpha_1^{(1)}\leq \beta_1^{(1)},\alpha_1^{(2)}\leq \beta_1^{(2)},\alpha_2^{(1)}\leq \beta_2^{(1)},\alpha_2^{(2)}\leq \beta_2^{(2)},\label{eqn:storage1}
\end{gather}
and if 
\begin{align}
\gamma_1^{(0)}-\beta_1^{(0)}+\gamma_1^{(1)}-\beta_1^{(1)}+\gamma_1^{(2)}-\beta_1^{(2)}< H(X_1)+H(X_2)+H(X_3)-H(X_0,X_1,X_2),
\end{align}
choose 
\begin{align}
(\alpha_1^{(0)},\alpha_1^{(1)},\alpha_1^{(2)},\gamma_1^{(0)},\gamma_1^{(1)},\gamma_1^{(2)})\in \mathcal{R}^*_{MD}\left(((V_1,V_2),X_0,X_1,X_2),\left(\{X_0,X_1,X_2\}\right)\right);
\end{align}
otherwise, choose $(\alpha_1^{(0)},\alpha_1^{(1)},\alpha_1^{(2)})=(\beta_1^{(0)},\beta_1^{(1)},\beta_1^{(2)})$. Similarly, if 
\begin{align}
\gamma_2^{(1)}-\beta_2^{(1)}+\gamma_2^{(2)}-\beta_2^{(2)}< I(Y_1;Y_2),
\end{align}
choose 
\begin{align}
(\alpha_2^{(1)},\alpha_2^{(2)},\gamma_2^{(1)},\gamma_2^{(2)})\in \mathcal{R}^*_{MD}\left(((V_1,V_2),Y_1,Y_2),\left(\{Y_1,Y_2\}\right)\right),
\end{align}
otherwise $(\alpha_2^{(1)},\alpha_2^{(2)})=(\beta_1^{(1)},\beta_1^{(2)})$;


\item The normalized average retrieval and storage rates 
\begin{align}
&2t\bar{\alpha}\geq \alpha_1^{(0)}+\alpha_1^{(1)}+\alpha_1^{(2)}+\alpha_2^{(1)}+\alpha_2^{(2)}\label{eqn:normalizedalphageneral},\\
&4t\bar{\beta}\geq 2\beta_1^{(0)}+ \beta_1^{(1)}+\beta_1^{(2)}+\beta_2^{(1)}+\beta_2^{(2)}.\label{eqn:normalizedbetageneral}
\end{align}
\end{enumerate}

Then we have the following theorem.
\begin{thm}
\label{theorem:innerboundgeneral}
$\mathcal{R}^{(t)}_{in}\subseteq \mathcal{R}$. 
\end{thm}

 We can in fact potentially enlarge the achievable region by taking $\cup_{t=1}^\infty\mathcal{R}^{(t)}_{in}$. However, unless $\mathcal{R}^{(t+1)}_{in}\subseteq \mathcal{R}^{(t)}_{in}$ for all $t\geq 1$, the region $\cup_{t=1}^\infty\mathcal{R}^{(t)}_{in}$ is even more difficult to characterize.  Nevertheless, for each fixed $t$, we can identify inner bounds by specifying a  feasible set of random variables $X_0,X_1,X_2,Y_1,Y_2$.

Instead of directly establishing this theorem, we shall prove the following theorem which establishes the existence of a PIR code with diminishing error probability, and then use an expurgation technique to extract a zero-error PIR code. 

\begin{thm}
\label{prop:epsilonerror}
Consider any $(\bar{\alpha},\bar{\beta})\in \mathcal{R}^{(t)}_{in}$. For any $\epsilon>0$ and sufficiently large $L$, there exists an $(L,L(\bar{\alpha}+\epsilon),L(\bar{\alpha}+\epsilon),L(\bar{\beta}+\epsilon),L(\bar{\beta}+\epsilon))$ $\epsilon$-error PIR code with the query distribution given in (\ref{eqn:Q1}) and (\ref{eqn:Q2}).
\end{thm}

The key observation to establish this theorem is that there are five descriptions in this setting, however, the retrieval and storage place different constraints on different combination of descriptions, and some descriptions can in fact be stored, recompressed, and then retrieved. Such compression and recompression may lead to storage savings. The description based on $X_0$ can be viewed as some common information to $X_1$ and $X_2$, which allows us to tradeoff the storage and retrieval rates.


\begin{proof}[Proof of Theorem \ref{prop:epsilonerror}]

\vspace{0.1cm}
\noindent\textit{Codebook generation:} Codebooks are built using the MD codebooks based on the distribution $((V_1,V_2),X_0,X_1,X_2,Y_1,Y_2)$.  

\vspace{0.1cm}
\noindent\textit{Storage codes: } The bin indices of the codebooks are stored in the two servers: those of $X_0$, $X_1$, and $X_2$ are stored at server-1 at rates $\alpha_1^{(0)}$, $\alpha_1^{(1)}$, and $\alpha_1^{(2)}$, respectively; those of $Y_1$ and $Y_2$ are stored at server-2 at rates $\alpha_2^{(1)}$ and $\alpha_2^{(2)}$. Note that at such rates, the codewords for $X_0$, $X_1$, and $X_2$ can be recovered  jointly with overwhelming probability, while those for $Y_1$ and $Y_2$ can also be recovered jointly with overwhelming probability.

\vspace{0.1cm}
\noindent\textit{Retrieval codes: } A different set of bin indices of the codebooks are retrieved during the retrieval process, again based on the MD$^*$ codebooks: those of $X_0$, $X_1$, and $X_2$ are retrieved at server-1 at rates $\beta_1^{(0)}$, $\beta_1^{(1)}$, and $\beta_1^{(2)}$, respectively; those of $Y_1$ and $Y_2$ are retrieved at server-2 at rates $\beta_2^{(1)}$ and $\beta_2^{(2)}$. Note that at such rates, the codewords of $X_0$, $X_1$, and $Y_1$ can be jointly recovered such that using the three corresponding codewords, the required $V_1$ source vector can be recovered with overwhelming probability. Similarly, the three retrieval patterns of $(X_0,X_1,Y_2)\rightarrow V_2$, $(X_0,X_2,Y_1)\rightarrow V_2$, and $(X_0,X_2,Y_2)\rightarrow V_2$ will succeed with overwhelming probabilities. 

\vspace{0.1cm} 
\noindent\textit{Storage and retrieval rates: } The rates can be computed straightforwardly, after normalization by the parameter $t$.
\end{proof}

Next we use it to prove Theorem \ref{theorem:innerboundgeneral}. 
\begin{proof}[Proof of Theorem \ref{theorem:innerboundgeneral}]
Given an $\epsilon>0$, according to Proposition \ref{prop:epsilonerror}, we can find an $(L,L(\bar{\alpha}+\epsilon),L(\bar{\alpha}+\epsilon),L(\bar{\beta}+\epsilon),L(\bar{\beta}+\epsilon))$ $\epsilon$-error PIR code for some sufficient large $L$. The probability of error of this code can be rewritten as
\begin{align*}
P_e&=0.5\sum_{k=1,2}\sum_{(w_1,w_2)} 2^{-2L} P_{Q^{[k]}_1,Q^{[k]}_2|(w_1,w_2)}(w_k\neq \hat{W}_k).
\end{align*}
For a fixed $(w_1,w_2)$ pair, denote the event that there exists a $(q_1,q_2)\in \{(11),(22)\}$, {\em i.e.,} when $(Q^{[1]}_1,Q^{[1]}_2)=(q_1,q_2)$, such that $\hat{w}_1\neq w_1$ as $E^{(1)}_{w_1,w_2}$, and there exists a $(q_1,q_2)\in \{(12),(21)\}$ such that $\hat{w}_2\neq w_2$ as $E^{(2)}_{w_1,w_2}$. Since $(Q^{[k]}_1,Q^{[k]}_2)$ is independent of $(W_1,W_2)$, if $P(E^{(k)}_{w_1,w_2})\neq 0$, we must have $P(E^{(k)}_{w_1,w_2})\geq 0.5$. It follows that
\begin{align}
P_e&\geq 0.25\sum_{(w_1,w_2)} 2^{-2L}  \mathbf{1}(E^{[1]}_{w_1,w_2}\cup E^{[2]}_{w_1,w_2}),
\end{align}
where $\mathbb(\cdot)$ is the indicator function. This implies that for any $\epsilon\leq 0.125$, there are at most $2^{2L-1}$ pairs of $(w_1,w_2)$ that will induce any coding error. We can use any $2^{2L-2}$ of the remaining $2^{2L-1}$ pairs of $L$-bit sequence pairs to instead store a pair of $(L-1)$-bit messages, through an arbitrary but fixed one-to-one mapping. This new code has a factor of $1+1/(L-1)$ increase in the normalized coding rates, which is negligible when $L$ is large. Thus a zero-error PIR code has been found with the same normalized rates as the $\epsilon$-error code asymptotically, and this completes the proof.
\end{proof}

\subsection{Outer bounds}
We next turn our attention to the outer bounds for $\mathcal{R}$, summarized in the following theorem.
\begin{thm}
\label{theorem:outerbounds}
Any $(\bar{\alpha},\bar{\beta})\in \mathcal{R}$ must satisfy
\begin{align}
\bar{\beta}\geq 0.75,\quad \bar{\alpha}+\bar{\beta}\geq 2,\quad 3\bar{\alpha}+8\bar{\beta}\geq 10.\label{eqn:ITouter}
\end{align}
Moreover, if $(\bar{\alpha},\bar{\beta})\in \mathcal{R}$ can be achieved by a linear code, it must satisfy
\begin{align}
\bar{\alpha}+6\bar{\beta}\geq 6. \label{bound:linear}
\end{align}
\end{thm}
The inequality $\bar{\beta}\geq 0.75$ follows from \cite{sun2017PIRcapacity}, while the two other bounds in (\ref{eqn:ITouter}) were proved in \cite{tian2020storage}. 
Therefore we only need to prove (\ref{bound:linear}). 



\begin{proof}[Proof of Theorem \ref{theorem:outerbounds}]
\label{sec:outerbounds}

Following  \cite{sun2018multiround}, we make the following simplifying assumptions that have no loss of generality. 
Define $\mathbb{Q} = \{Q_1^{[1]}, Q_1^{[2]}, Q_2^{[1]}, Q_2^{[2]}\}$.
\begin{eqnarray}
&& 1. ~~Q_1^{[1]} = Q_1^{[2]} ~\Rightarrow~ A_1^{[1]} = A_1^{[2]} \label{eq:1same},\\
&& 2. ~~H(A_1^{[1]} | \mathbb{Q}) = H(A_2^{[1]} | \mathbb{Q}) = H(A_2^{[2]} | \mathbb{Q}),\quad H(S_1) = H(S_2) \label{eq:sym0}\\
&& ~~~\Rightarrow ~H(A_1^{[1]} | \mathbb{Q}) \leq \beta \leq (\bar{\beta} + \epsilon)L , \quad H(S_2) \leq \alpha \leq (\bar{\alpha} + \epsilon)L. \label{eq:sym} 
\end{eqnarray}
Assumption 1 states that the query to the first database is the same regardless of the desired message index. This is justified by the privacy condition that the query to one database is independent of the desired message index. Assumption 2 states that the scheme is symmetric after the symmetrization operation in Lemma 1 (the proof is referred to Theorem 3 in \cite{sun2018multiround}). (\ref{eq:sym}) follows from the fact that to describe $S_2, A_1^{[1]}$, the number of bits needed can not be less than the entropy value, and Definition 1.

In the following, we use $(c)$ to refer to the correctness condition, $(	i)$ to refer to the constraint that queries are independent of the messages, $(a)$ to refer to the constraint that answers are deterministic functions of the storage variables and corresponding queries, and $(p)$ to refer to the privacy condition.

%
From $A_1^{[1]}, A_2^{[1]}, \mathbb{Q}$, we can decode $W_1$.
\begin{eqnarray}
H(A_1^{[1]}, A_2^{[1]} | W_1,\mathbb{Q}) &=& H(A_1^{[1]}, A_2^{[1]}, W_1 | \mathbb{Q}) - H(W_1 | \mathbb{Q}) \\
&\overset{(c)(i)}{=}& H(A_1^{[1]}, A_2^{[1]} | \mathbb{Q}) - L 
\\
&\overset{(\ref{eq:sym0})}{\leq}& 2 H(A_1^{[1]} | \mathbb{Q}) - L. 
\label{eq:int}
\end{eqnarray}

Next, consider Ingleton's inequality.
\begin{eqnarray}
 I(A_2^{[1]} ; A_2^{[2]} | \mathbb{Q}) &\leq& I(A_2^{[1]} ; A_2^{[2]} | W_1, \mathbb{Q}) + I(A_2^{[1]} ; A_2^{[2]} | W_2, \mathbb{Q}) \\
&=& 2 I(A_2^{[1]} ; A_2^{[2]} | W_1, \mathbb{Q}) \label{eq:switch} \\
&=& 2 \big(H(A_2^{[1]} | W_1, \mathbb{Q}) + H(A_2^{[2]}|W_1, \mathbb{Q}) - H(A_2^{[1]}, A_2^{[2]} | W_1,\mathbb{Q}) \big) \\
&\overset{(p)}{=}& 2 \big(2H(A_2^{[1]} | W_1, \mathbb{Q}) - H(A_2^{[1]}, A_2^{[2]} | W_1,\mathbb{Q}) \big) \label{eq:p} \\
&\leq& 2 \big(2H(A_2^{[1]} | W_1, \mathbb{Q}) + H(A_1^{[1]}, A_2^{[1]} | W_1,\mathbb{Q})  \notag\\
&&-~ H(A_1^{[1]}, A_2^{[1]}, A_2^{[2]} | W_1, \mathbb{Q}) - H(A_2^{[1]} | W_1, \mathbb{Q}) \big) \label{eq:sub} \\
&\overset{(c)(\ref{eq:1same})}{=}& 2 \big( H(A_2^{[1]} | W_1, \mathbb{Q}) + H(A_1^{[1]}, A_2^{[1]} | W_1,\mathbb{Q})  \notag\\
&&-~ H(A_1^{[1]}, A_2^{[1]}, A_2^{[2]}, W_2 | W_1, \mathbb{Q}) 
\big) \\
&\overset{(i)}{\leq}& 2 \big( 2 H(A_1^{[1]}, A_2^{[1]} | W_1,\mathbb{Q}) - H(W_2) 
\big) \\
&\overset{(\ref{eq:int})}{\leq}& 2 \big(2 (2H(A_1^{[1]} | \mathbb{Q}) - L) - L 
\big) \label{eq:i1}
\end{eqnarray}
where (\ref{eq:switch}) follows from the observation that the second term can be bounded using the same method as that bounds the first term by switching the message index. A more detailed derivation of (\ref{eq:p}) appears in (79) of \cite{sun2018multiround}. (\ref{eq:sub}) is due to sub-modularity of entropy.

Note that 
\begin{eqnarray}
I(A_2^{[1]} ; A_2^{[2]} | \mathbb{Q}) &=& H(A_2^{[1]} | \mathbb{Q}) + H(A_2^{[2]} | \mathbb{Q}) - H(A_2^{[1]}, A_2^{[2]} | \mathbb{Q}) \\
&\overset{(\ref{eq:sym0})}{\geq}& 2H(A_1^{[1]} | \mathbb{Q})  - (\bar{\alpha} + \epsilon)L  \label{eq:i2}
\end{eqnarray}
where in (\ref{eq:i2}), and the second term is bounded as follows :
\begin{eqnarray}
 H(A_2^{[1]}, A_2^{[2]}|\mathbb{Q}) \leq  H(A_2^{[1]}, A_2^{[2]}, S_2|\mathbb{Q}) \overset{(a)}{=} H(S_2|\mathbb{Q}) 
\overset{(\ref{eq:sym})}{\leq} (\bar{\alpha} + \epsilon)L.
\label{eq:f1}
\end{eqnarray}

Combining (\ref{eq:i1}) and (\ref{eq:i2}), we have
\begin{eqnarray}
&& 2H(A_1^{[1]} | \mathbb{Q})/L - (\bar{\alpha}+\epsilon) \geq 2(4H(A_1^{[1]} | \mathbb{Q})/L - 3) \notag\\
&\Rightarrow& \bar{\alpha}+\epsilon + 6 H(A_1^{[1]} | \mathbb{Q})/L \geq 6 \\
&\overset{(\ref{eq:sym})}{\Rightarrow}& \bar{\alpha} + 6 \bar{\beta} \geq 6.
\end{eqnarray}
The proof is complete.
\end{proof}

\subsection{Specialization of the Inner Bound}

The inner bound given in Theorem \ref{theorem:innerboundgeneral} is general but more involved, and we can specialize it in multiple ways in order to simplify it. One particularly interesting approach is as follows. Define the region $\tilde{\mathcal{R}}^{(t)}_{in}$ to be the collection of $(\bar{\alpha},\bar{\beta})$ pairs such that there exists random variables $(X_0,X_1,X_2,Y_1,Y_2)$ jointly distributed with $(V_1,V_2)$ such that 
\begin{enumerate}
\item The distribution factorizes as follows $$P_{V_1,V_2,X_0,X_1,X_2,Y_1,Y_2}=P_{V_1,V_2}P_{X_0|V_1,V_2}P_{X_1|V_1,V_2}P_{X_2|V_1,V_2}P_{Y_1|V_1,V_2}P_{Y_2|V_1,V_2};$$
\item There exist deterministic functions $f_{1,1}$, $f_{1,2}$, $f_{2,1}$, and $f_{2,2}$ such that
\begin{gather}
V_1=f_{1,1}(X_0,X_1,Y_1)=f_{2,2}(X_0,X_2,Y_2)\label{eqn:decoding1},\\
V_2=f_{1,2}(X_0,X_1,Y_2)=f_{2,1}(X_0,X_2,Y_1)\label{eqn:decoding4};
\end{gather}
\item A set of rates 
\begin{gather}
\gamma_1^{(0)}=I(V_1,V_2;X_0),\,\gamma_1^{(1)}=I(V_1,V_2;X_1),\,\gamma_1^{(2)}=I(V_1,V_2;X_2),\label{eqn:gammas}\\
\gamma_2^{(1)}=I(V_1,V_2;Y_1),\,\gamma_2^{(2)}=I(V_1,V_2;Y_2),\\
\beta_1^{(0)}=\gamma_1^{(0)},\,\beta_1^{(1)}=I(V_1,V_2;X_1|X_0),\,\beta_1^{(2)}=I(V_1,V_2;X_2|X_0),\\
\beta_2^{(1)}=\max(I(V_1,V_2;Y_1|X_0,X_1),I(V_1,V_2;Y_1|X_0,X_2)),\\
\beta_2^{(2)}=\max(I(V_1,V_2;Y_2|X_0,X_1),I(V_1,V_2;Y_2|X_0,X_2)),\label{eqn:betas}
\end{gather}
and $
(\alpha_1^{(0)}=\gamma_1^{(0)},\alpha_1^{(1)},\alpha_1^{(2)}, \alpha_2^{(1)},\alpha_2^{(2)})$ as defined in item 3 for the general region $\mathcal{R}^{(t)}$;
\item The normalized average retrieval and storage rates 
\begin{align}
&2t\bar{\alpha}\geq \alpha_1^{(0)}+\alpha_1^{(1)}+\alpha_1^{(2)}+\alpha_2^{(1)}+\alpha_2^{(2)},\label{eqn:normalizedalpha}\\
&4t\bar{\beta}\geq 2\beta_1^{(0)}+ \beta_1^{(1)}+\beta_1^{(2)}+\beta_2^{(1)}+\beta_2^{(2)}.\label{eqn:normalizedbeta}
\end{align}
\end{enumerate}
Then we have the following corollary.
\begin{corollary}
\label{theorem:innerbound}
$\tilde{\mathcal{R}}^{(t)}_{in}\subseteq \mathcal{R}$. 
\end{corollary}

This inner bound is illustrated together with the outer bounds in Fig. \ref{fig:bounds}.

\begin{figure}
\centering
\includegraphics[width=0.65\textwidth]{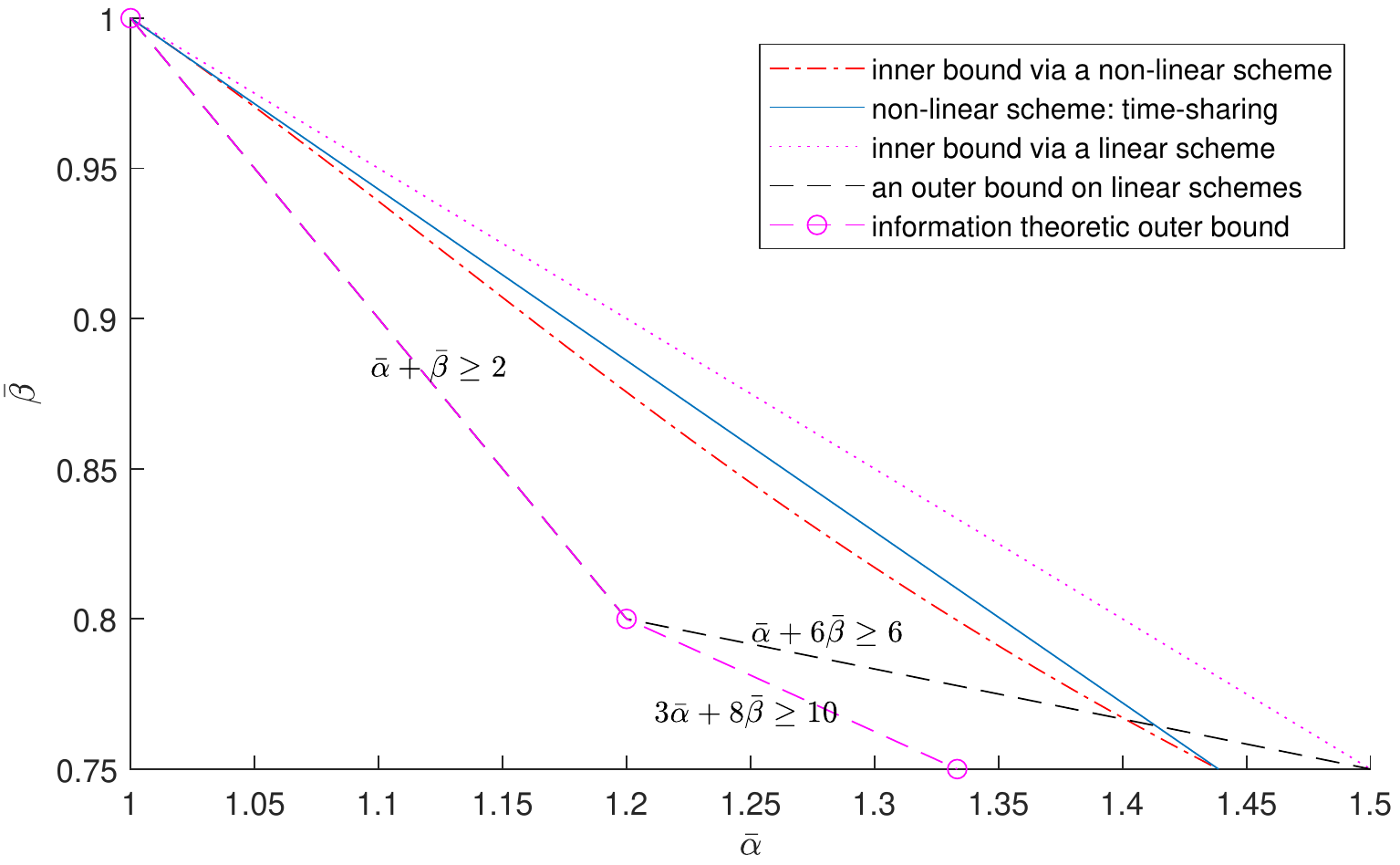}
\caption{Illustration of inner bounds and outer bounds.\label{fig:bounds}}
\end{figure}

\begin{proof}
The main difference from Theorem \ref{theorem:innerboundgeneral} is in the special dependence structure of $(X_0,X_1,X_2,Y_1,Y_2)$ jointly distributed with $(V_1,V_2)$, i.e., the Markov structure. 
We verify that the rate assignments satisfy all the constraints in Theorem \ref{theorem:innerboundgeneral}. 
Due to the special dependence structure of $(X_0,X_1,X_2,Y_1,Y_2)$ jointly distributed with $(V_1,V_2)$, it is straightforward to verify that $$(\gamma_1^{(0)},\gamma_1^{(1)},\gamma_1^{(2)},\gamma_2^{(1)},\gamma_2^{(2)})\in \mathcal{R}_{MD}((V_1,V_2),X_0,X_1,X_2,Y_1,Y_2).$$ 
We next verify (\ref{eqn:10R}) holds with the choice given above. Due to the symmetry in the structure, we only need to confirm one subset of random variables, i.e., $\{X_0,X_1,Y_1\}$, and the three other subsets $\{X_0,X_1,Y_2\}$,  $\{X_0,X_2,Y_1\}$, and  $\{X_0,X_2,Y_2\}$ follow similarly. There are a total of $7$ conditions in the form of (\ref{eqn:rates2}) associated with this subset $\{X_0,X_1,Y_1\}$. Notice that
\begin{align*}
\gamma_1^{(0)}-\beta_1^{(0)}=0,\,\gamma_1^{(1)}-\beta_1^{(1)}=I(X_1;X_0),\,\gamma_2^{(2)}-\beta_2^{(2)}\leq I(Y_1;X_0,X_1),
\end{align*}
which in fact confirm three of the seven conditions when $\mathcal{J}$ is a singleton. Next when $\mathcal{J}$ has two elements, we verify that
\begin{align}
\gamma_1^{(0)}-\beta_1^{(0)}+\gamma_1^{(1)}-\beta_1^{(1)}&=I(X_1;X_0)= H(X_0)+H(X_1)-H(X_0,X_1)\notag\\
&\leq H(X_0)+H(X_1)-H(X_0,X_1|Y_1),\\
\gamma_1^{(0)}-\beta_1^{(0)}+\gamma_2^{(1)}-\beta_2^{(1)}&\leq I(Y_1;X_0,X_1)=H(Y_1)+H(X_0,X_1)-H(X_0,X_1,Y_1)\notag\\
&\leq H(Y_1)+H(X_0)+H(X_1)-H(X_0,X_1,Y_1)\notag\\
&= H(X_0)+H(Y_1)-H(X_0,Y_1|X_1),\\
\gamma_1^{(1)}-\beta_1^{(1)}+\gamma_2^{(1)}-\beta_2^{(1)}&\leq I(X_1;X_0)+I(Y_1;X_0,X_1)=H(X_1)+H(Y_1)-H(X_1,Y_1|X_0).
\end{align}
Finally when $\mathcal{J}$ has all the three elements, we have
\begin{align}
&\gamma_1^{(0)}-\beta_1^{(0)}+\gamma_1^{(1)}-\beta_1^{(1)}+\gamma_2^{(1)}-\beta_2^{(1)}\notag\\
&=I(X_0;X_1)+I(V_1,V_2;X_1)-\max(I(V_1,V_2;Y_1|X_0,X_1),I(V_1,V_2;Y_1|X_0,X_2))\\
&\leq I(X_0;X_1)+I(V_1,V_2;X_1)-I(V_1,V_2;Y_1|X_0,X_1)\\
&=H(X_0)+H(X_1)+H(Y_1)-H(X_0,X_1,Y_1).
\end{align}
Thus (\ref{eqn:10R}) is indeed true with the assignments (\ref{eqn:gammas})-(\ref{eqn:betas}). This in fact completes the proof. 
\end{proof}


\begin{table}[t!]
\caption{Conditional distribution $P_{X_0|W_1,W_2}$ used in Corollary \ref{coro:testchannel}. \label{tab:distribution}}
\[\arraycolsep=1.5pt\def\arraystretch{1.0}
\begin{array}{|c|cccc|}\hline
           (w_1,w_2)         &    x_0=(00)  & x_0=(01)   & x_0=(10) & x_0=(11)   \\\hline
(00)&   1/2            &                  &                 & 1/2\\
(10)&   (1-p)/2      &     p           &                 & (1-p)/2\\
(01)&   (1-p)/2      &                  &    p           & (1-p)/2\\
(11)&   1/2            &                  &                 & 1/2\\\hline
\end{array}\nonumber\]
\end{table}

We can use any explicit distribution $(X_0,X_1,X_2,Y_1,Y_2)$ to obtain an explicit inner bound to $\tilde{\mathcal{R}}^{(t)}_{in}$, and the next corollary provides one such non-trivial bound. For convenience, we write the entropy function of a probability mass   $(p_1,\ldots,p_t)$ as $H(p_1,\ldots,p_t)$.
\begin{corollary}
\label{coro:testchannel}
The following $(\bar{\alpha},\bar{\beta})\in \mathcal{R}$ for any $p\in [0,1]$:
\begin{align*}
\bar{\alpha}=&\frac{9}{4}-H(\frac{1}{4},\frac{3}{4})+\frac{1}{4}H(\frac{1-p}{2},\frac{1-p}{2},\frac{p}{2},\frac{p}{2})\\
&+\frac{1}{2}H(\frac{2-p}{4},\frac{2-p}{4},\frac{p}{2})-\frac{3}{4}H(\frac{3-2p}{6},\frac{3-2p}{6},\frac{p}{3},\frac{p}{3}),\\
\bar{\beta}=&\frac{5}{8}+\frac{1}{4}H(\frac{2-p}{4},\frac{2-p}{4},\frac{p}{2})-\frac{1}{8}H(\frac{1-p}{2},\frac{1-p}{2},p).
\end{align*}
\end{corollary}
\begin{proof}
These tradeoff pairs are obtained by applying Corollary \ref{theorem:innerbound}, and choosing $t=1$ and setting $(X_1,X_2,Y_1,Y_2)$ as given in (\ref{eqn:distribution}), and letting $X_0$ be defined as in Table \ref{tab:distribution}. Note that the joint distribution indeed satisfies the required Markov structure, and in this case $\alpha_2^{(1)}=\beta_2^{(1)}$ and $\alpha_2^{(2)}=\beta_2^{(2)}$.
\end{proof}

\section{Conclusion}
We consider the problem of private information retrieval using a Shannon-theoretic approach. A new coding scheme based on random coding and binning is proposed, which reveals a hidden connection to the multiple description problem. It is shown that for the $(2,2)$ PIR setting, this non-linear coding scheme is able to provide the best known tradeoff between retrieval rate and storage rate, which is strictly better than that achievable using linear codes. We further investigate the relation between zero-error PIR codes and $\epsilon$-error PIR codes in this setting, and shows that they do not causes any essential difference in this problem setting. We hope that the hidden connection to multiple description coding can provide a new revenue to design more efficient PIR codes.

\bibliographystyle{IEEEtran}

\end{document}